\pdfoutput=1
\documentclass{article}

\usepackage{PRIMEarxiv}
 
\usepackage[utf8]{inputenc}
\usepackage[T1]{fontenc}
\usepackage{hyperref}
\usepackage{url}
\usepackage{booktabs}
\usepackage{amsmath,amssymb,amsfonts,amsthm}
\usepackage{nicefrac}
\usepackage{microtype}
\usepackage{fancyhdr}
\usepackage{graphicx}
\usepackage{xcolor}
\usepackage{colortbl}
\usepackage{enumitem}
\usepackage{subcaption}
\usepackage{adjustbox}
\graphicspath{{media/}}
 
\newtheorem{proposition}{Proposition}
\theoremstyle{definition}
\newtheorem{definition}{Definition}
\newtheorem{remark}{Remark}
\newtheorem{example}{Example}
\newcommand{\C}{\mathbb{C}}

\newcommand\blfootnote[1]{%
  \begingroup
  \let\thefootnote\relax\footnotetext{#1}%
  \endgroup
}
\title{Spectral Braids: A Topological Descriptor for Dynamic Directed Networks}
 
\author{
  Mehmet Emin Aktas \\
  School of Data Science and Analytics \\
  Kennesaw State University \\
  Kennesaw, GA, USA \\
  \texttt{maktas1@kennesaw.edu} \\
  \And
  Esra Akbas \\
  Department of Computer Science \\
  Georgia State University \\
  Atlanta, GA, USA \\
  \texttt{eakbas1@gsu.edu} \\
}
 
\begin{document}
\maketitle
\blfootnote{Accepted at the 2026 IEEE International Conference on Big Data (IEEE BigData 2026). DOI: TBD.\newline
\copyright~2026 IEEE. Personal use of this material is permitted. Permission from IEEE must be obtained for all other uses, in any current or future media, including reprinting/republishing this material for advertising or promotional purposes, creating new collective works, for resale or redistribution to servers or lists, or reuse of any copyrighted component of this work in other works.}

\begin{abstract}
Eigenvalues are a workhorse of graph machine learning: spectral clustering,
spectral graph neural networks, and graph kernels all summarize a network by the
eigenvalues of a matrix. These eigenvalues are almost always treated as a static list
of numbers. When a graph changes, however, its eigenvalues move, and when the matrix
is not symmetric, they move in the complex plane, where they can wind around each
other. The pattern of this winding is a \emph{braid}, a classical object from
topology. We propose \emph{spectral braids}, a methodology for dynamic directed
networks: each snapshot is aggregated into a community-level transition matrix, its complex eigenvalues are tracked over time,
and the braid they trace is read off as a word. The word is then mapped to a matrix
by the Burau representation, which gives a fixed-size feature vector for any
classifier. The braid has a concrete meaning: each eigenvalue describes how a group
of communities behaves, and each crossing records one group overtaking another. We
evaluate the method on dynamic network classification, using synthetic networks of
online communities whose labels describe how the network evolved between the same
starting and ending states. Our experimental results show that the proposed features
reach 97.1\% accuracy on four classes, against 44.5\% for the best feature set that
does not track eigenvalues, and that the Burau features also separate networks
containing the same events in a different order, which simple counts of crossings
cannot do. To the best of our knowledge, this is the first study that defines braids
and braid monodromy on graph data and uses them as features for learning on dynamic
networks.
\end{abstract}


\section{Introduction}
Most spectral methods for graphs share one step: build a matrix from the graph (adjacency, Laplacian or transition matrix), compute its eigenvalues, and use them as features or as a basis for clustering and learning. The eigenvalues are treated as a \emph{set of numbers}. That works well for a single graph. Many graphs of interest, however, change. Communication platforms record who messages whom, payment systems who pays whom, and mobility systems which regions trips connect, continuously and at a massive scale. As the graph changes, its eigenvalues move.

For symmetric matrices the eigenvalues are real numbers. They move along a line, and two eigenvalues on a line can never pass each other without colliding. For \emph{non-symmetric} matrices, which arise naturally from directed graphs, the eigenvalues can be complex numbers. They move in a \emph{plane}, where they can go around each other. A sequence of such go-arounds is exactly what topologists call a \emph{braid} (Fig.~\ref{fig:braid}). Braids are the natural language for relative motion in the plane. They have a long history in topology and algebraic geometry, where \emph{braid monodromy} describes how the roots of a polynomial wind around each other as a parameter moves around a loop~\cite{zariski1929,moishezon1981,aktas2019braid}. Braids are also used to analyze particle trajectories in fluids~\cite{thiffeault2010}, and braided complex eigenvalues were recently observed in physics experiments~\cite{wang2021braiding}.

In this paper, we introduce \textit{spectral braids}, the braid monodromy for dynamic directed networks. We first aggregate each snapshot of the network into a small community-level transition matrix. Because edges are directed, this matrix is non-symmetric, and its eigenvalues can be complex. We then follow these eigenvalues from snapshot to snapshot and record how they pass around each other as a braid word. The network need not return to its starting state: any observation window yields a braid, and a network that does return yields the braid monodromy in the classical sense. To use the braid in machine learning, we map the word to a matrix with the Burau representation~\cite{burau1936}, which turns braids of any length into feature vectors of a fixed size. We show that the braid is non-trivial only when the network carries \emph{circulation}, that is, flow that goes around cycles more in one direction than the other. For real matrices, we introduce an \emph{upper-half-plane braid} that keeps the direction of rotation. 

We also give the eigenvalues a concrete meaning. Consider a group of online
communities that pass content among themselves. If all interactions stayed inside
each community, the transition matrix of the group would be the identity, and its
eigenvalues would equal one; the more interactions stay inside, the further to the
right the eigenvalue of the group lies. Its real part, therefore, measures how much
the group talks to itself, which we call its echo-chamber level. If the interactions
between the communities were balanced, with as many going from $a$ to $b$ as from
$b$ to $a$, the matrix would be symmetric and its eigenvalues real. An imaginary
part appears only when content flows around the group more in one direction than in
the other, and it grows with this imbalance. The imaginary part, therefore, measures
how one-way the interactions are, which we call the broadcast level of the group
(Section~\ref{sec:interp} gives the exact formula). With this reading, a crossing in
the braid has a plain meaning. Two eigenvalues cross when their real parts exchange
order, which is the moment one group becomes a stronger echo chamber than the other.
The sign of the crossing records which eigenvalue passes above the other, that is,
which of the two groups was broadcasting more at that moment.

Next, we build two datasets of dynamic networks labeled by how their strongest echo chamber formed. Networks with different labels begin and end in the same state, so only the path distinguishes them. Our experimental results show that Burau features identify the label in 97.1\% of the networks at about 3{,}400 edges per snapshot, where features that ignore the path stay at chance and the best features that do not track eigenvalues reach 44.5\%. When networks contain the same events in a different order, counts of crossings fall to chance, and Burau features reach 90.8\%. To our knowledge, this is the first use of braids and braid monodromy for graph data.

The main contributions of this work are summarized as follows.
\begin{itemize}[leftmargin=*,itemsep=1pt,topsep=2pt]
\item We propose \emph{spectral braids}: a pipeline from edge streams to braid words
that is linear in the number of edges and is defined for any observation window, not
only for periodic networks (Section~\ref{sec:spec}).
\item We describe the tracked eigenvalues in two ways for dynamic network
classification: by geometric statistics of the strands, and by the braid, which the
Burau representation turns into fixed-size features (Section~\ref{sec:burau}).
\item We characterize when spectral braids are informative. They are trivial without
circulation, and for real matrices, the braid of all eigenvalues consists of the braid of the upper
half-plane and its mirror image, which motivates the \emph{upper-half-plane braid}. We also give the
eigenvalues and the crossings a concrete meaning
(Sections~\ref{sec:props}--\ref{sec:interp}).
\item We introduce two synthetic datasets of dynamic networks of online communities,
with meaningful labels and known ground truth (Section~\ref{sec:bench}), and evaluate the method on them for dynamic network classification: accuracy against spectral
baselines, transfer to networks that look different from the training data,
sensitivity to the number of edges per snapshot, and running time
(Section~\ref{sec:exp}).
\item We report four real datasets on which the method does not apply and explain
why (Section~\ref{sec:scope}).
\end{itemize}
\textbf{Reproducibility:} The source code for the proposed methods and the datasets is available at \url{https://github.com/maktas3838/spectral-braids}.
\section{Background and Related Work}

\subsection{Dynamic Directed Networks}
\label{sec:prelim}
A \emph{directed graph} is a pair $G=(V,E)$, where $V$ is a finite set of nodes and
$E\subseteq V\times V$ is a set of directed edges; an edge $(u,v)$, written $u\to v$,
goes from $u$ to $v$. The graph is \emph{undirected} if $u\to v\in E$ whenever
$v\to u\in E$. Edges may carry positive weights, in which case every count of edges
below is replaced by a sum of weights.

\begin{definition}[Dynamic directed network]\label{def:dyn}
A \emph{dynamic directed network} is a sequence $G=(G_0,\dots,G_T)$ of directed
graphs $G_t=(V,E_t)$ on a common node set $V$. Each $G_t$ is called a
\emph{snapshot} and contains the edges observed in the $t$-th time window.
\end{definition}

Furthermore, a \emph{community partition} of $G$ is a partition
$\mathcal{C}=\{C_1,\dots,C_K\}$ of $V$ into $K$ disjoint communities. In this paper, we also examine how communities in dynamic directed networks evolve over time through the community transition matrices defined as follows. 

\begin{definition}[Community transition matrix]\label{def:ctm}
Let $G$ be a dynamic directed network with community partition $\mathcal{C}$, and
let
\[
N_t(a,b)=\#\{u\to v\in E_t:\ u\in C_a,\ v\in C_b\}
\]
be the number of edges of $G_t$ from $C_a$ to $C_b$. The \emph{community transition
matrix} of the snapshot $G_t$ is the $K\times K$ matrix
\begin{equation}\label{eq:ctm}
P_t(a,b)=\frac{N_t(a,b)}{\sum_{c=1}^{K}N_t(a,c)},
\end{equation}
the fraction of the edges leaving $C_a$ that arrive in $C_b$. 
\end{definition}

$P_t$ is row-stochastic. It is the transition matrix of a random walk that
moves among communities by following a random edge~\cite{piccardi2011}, and its
diagonal entry $P_t(a,a)$ is the fraction of edges that stay inside $C_a$. Since
edges are directed, $N_t(a,b)\neq N_t(b,a)$ in general, so $P_t$ is not symmetric, and its eigenvalues can be complex.

\begin{figure}[t]
\centering
\includegraphics[width=\textwidth]{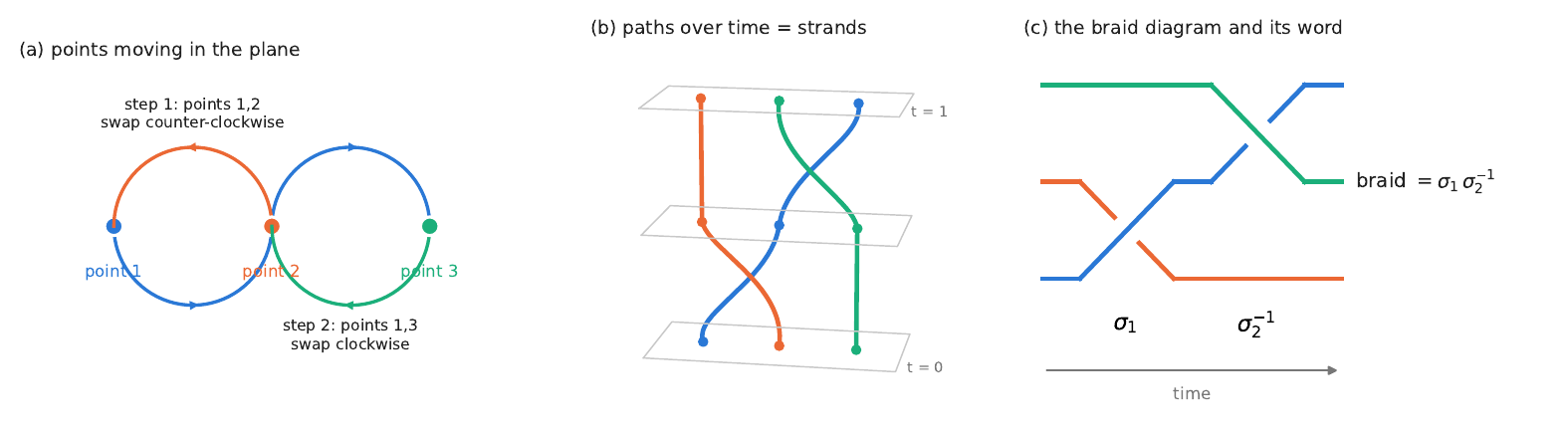}

\caption{Braids. (a) Points moving in the plane: points 1 and 2 swap
counter-clockwise, then points 1 and 3 swap clockwise. (b) Stacked over time, the
paths form strands. (c) Viewed from the front, each swap of neighboring strands is a
crossing, written $\sigma_i$ (counter-clockwise) or $\sigma_i^{-1}$ (clockwise); this
motion is the braid $\sigma_1\sigma_2^{-1}$.}
\label{fig:braid}
\end{figure}

\subsection{Braids and Braid Monodromy}
Consider $n$ distinct points moving in the plane over a time
interval. Stacking their positions over time produces $n$ non-intersecting strands
(Fig.~\ref{fig:braid}a,b).

\begin{definition}[Braid]\label{def:braid}
Fix $n$ distinct points $p_1,\dots,p_n\in\C$. A \emph{braid on $n$ strands} is a
collection of $n$ disjoint paths (strands) in $\C\times[0,1]$ such that each strand
meets every plane $\C\times\{s\}$ exactly once, the $i$-th strand starts at
$(p_i,0)$, and the strands end at the points $(p_1,1),\dots,(p_n,1)$ in some order.
Equivalently, a braid is a path of sets of $n$ distinct points in $\C$ that starts and ends at $\{p_1,\dots,p_n\}$, up to deformation.
\end{definition}

Projecting the strands onto a line and recording every
exchange of neighboring strands gives a word in the generators
$\sigma_1,\dots,\sigma_{n-1}$ and their inverses (Fig.~\ref{fig:braid}c). Two words
represent the same braid if the strands can be deformed into each other without
passing through one another. For example, $\sigma_1\sigma_1^{-1}$ is the trivial
braid, and $\sigma_1\sigma_2\sigma_1=\sigma_2\sigma_1\sigma_2$. Braids on $n$ strands
form the braid group $B_n$, where the product of two braids is obtained by placing
one after the other. The order of crossings matters:
$\sigma_1\sigma_2\neq\sigma_2\sigma_1$.

\emph{Braid monodromy} is the braid traced by the roots of a polynomial whose
coefficients depend on a parameter. Let $f(x,y)$ be a polynomial of degree $n$ in
$y$, and let $\Delta\subset\C$ be the finite set of values of $x$ at which $f(x,\cdot)$
has a repeated root. For every $x\notin\Delta$ the equation $f(x,y)=0$ has $n$
distinct roots. When $x$ moves along a loop in $\C\setminus\Delta$, the roots move
without colliding and return to the set they started from, so they trace a braid.
This assigns a braid to every loop and defines a homomorphism
$\pi_1(\C\setminus\Delta)\to B_n$, the braid monodromy of
$f$~\cite{zariski1929,moishezon1981,aktas2019braid}. In this paper, the polynomial is
the characteristic polynomial $\det(\lambda I-P_t)$ of a matrix built from the
network, its roots are the eigenvalues, and the parameter is time.

Furthermore, the Burau representation~\cite{burau1936} maps every braid to an
$n\times n$ matrix that depends on a parameter $t$. The generator $\sigma_i$ is sent
to the identity matrix with the block
$\bigl(\begin{smallmatrix}1-t & t\\ 1 & 0\end{smallmatrix}\bigr)$ placed at rows and
columns $i,i+1$, and a word is sent to the product of the matrices of its letters.
Two words of the same braid give the same matrix. The representation is faithful for
$n\le 4$ and not for $n\ge5$~\cite{bigelow1999}.

\subsection{Related Work}
\subsubsection{Learning on dynamic graphs} Methods for dynamic graphs either process a
sequence of snapshots or a stream of timed edges~\cite{kazemi2020}. Temporal graph
neural networks combine message passing with a model of time, for example by
evolving the weights of a graph convolution~\cite{pareja2020}, by attention over
structure and time~\cite{sankar2020,xu2020}, or by a memory per
node~\cite{rossi2020}. These models learn representations end-to-end and do not produce an explicit description of how the network's structure changed.

\subsubsection{Spectral methods for dynamic graphs} A second line of work summarizes each
snapshot by its spectrum. Graph descriptors built from Laplacian eigenvalues compare
graphs of different sizes~\cite{verma2017,tsitsulin2018}. For dynamic graphs,
eigenvectors or singular values of consecutive snapshots are compared to detect
anomalies and change points~\cite{ide2004,huang2020}, spectral clustering is
smoothed over time~\cite{chi2007}, and eigenpairs are updated incrementally as edges
arrive~\cite{chen2015}. Graph signal processing extends the graph Fourier transform
to signals that vary in time~\cite{grassi2018}. In all of these, the spectrum of a
snapshot is treated as a set or a sorted list, and the matrices are symmetric, so
the eigenvalues are real. None of them records how eigenvalues move relative to each
other, which is the information a braid carries.

\subsubsection{Spectral methods for directed graphs} Directed graphs give non-symmetric
matrices. The usual response is to replace them by a symmetric or Hermitian
operator, such as the Laplacian of Chung~\cite{chung2005}, the Hermitian adjacency
matrix~\cite{guo2017}, or the magnetic Laplacian~\cite{fanuel2018}, which is also
the basis of spectral neural networks for directed graphs~\cite{zhang2021magnet}.
These operators have real eigenvalues by design. We take the opposite route and keep
the non-symmetric matrix, because its complex eigenvalues are what can braid.
Aggregating a network into a Markov chain on its communities is a standard tool for
testing community structure~\cite{piccardi2011}; there, the probability of staying in a community is called its persistence probability. Irreversibility of such
chains is often measured by entropy production~\cite{lynn2021}, a single number that
does not describe \emph{how} the dynamics circulate.

\subsubsection{Topology and braids} Topological data analysis of graphs relies mostly on
persistent homology~\cite{aktas2019ph}, which has been applied to time-varying
graphs~\cite{hajij2018,myers2023}; vineyards~\cite{cohensteiner2006} track
persistence diagrams through time. Braids have been used outside topology to
quantify mixing from particle trajectories in fluids~\cite{thiffeault2010}, and
braids of complex eigenvalues have been observed in non-Hermitian physical
systems~\cite{wang2021braiding,patil2022}. To our knowledge, braids have not been
defined on graph data.

\section{Methodology}\label{sec:method}
In this section, we present our methodology for describing a dynamic directed
network by a braid. We first define the spectral braid of a dynamic network with
respect to a matrix representation, explain why we choose the community transition
matrix, and give the three steps that compute the braid from the snapshots
(Section~\ref{sec:spec}). We then turn the braid into features for machine learning:
the Burau matrix, which gives a fixed-size vector for braids of any length, and the
tracked geometry of the strands (Section~\ref{sec:burau}). Next, we establish the
theoretical properties of spectral braids, namely when the braid is trivial, why we
restrict it to the upper half-plane, and what it is invariant to
(Section~\ref{sec:props}). Finally, we explain what the eigenvalues and the
crossings of the braid mean in terms of the behaviour of groups of communities
(Section~\ref{sec:interp}).

\subsection{Spectral Braids}\label{sec:spec}
\begin{figure}[t]
\centering
\includegraphics[width=\textwidth]{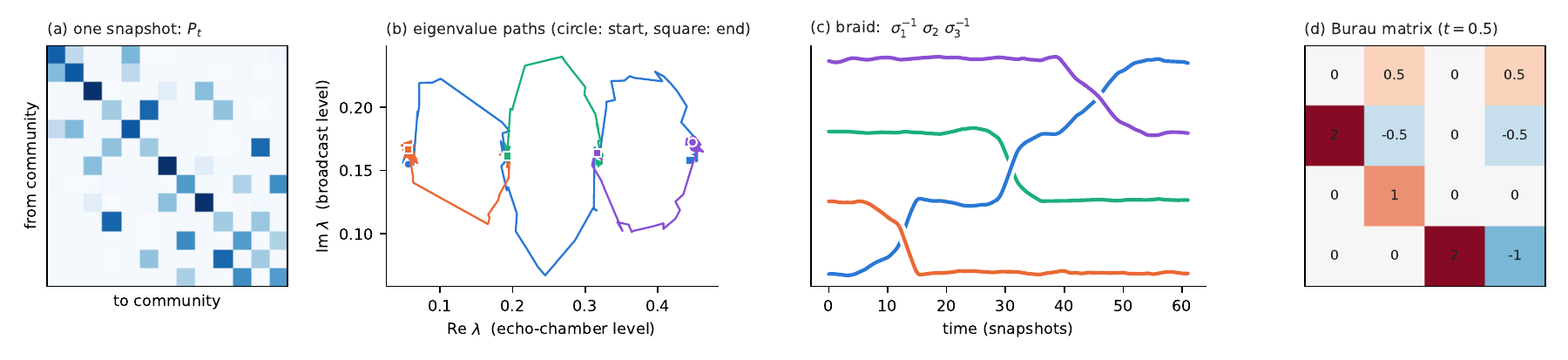}
\caption{Spectral braids on one dynamic network with 13 communities. (a) The
community transition matrix of one snapshot. (b) The paths of the four eigenvalues in
the upper half-plane over about 60 snapshots. The blue eigenvalue moves from the left end
to the right end and passes the three others. (c) The same paths drawn over time: the
blue strand crosses the three others one by one, which gives the braid
$\sigma_1^{-1}\sigma_2\sigma_3^{-1}$. (d) The Burau matrix of this braid.}
\label{fig:pipeline}
\end{figure}

Let $G$ be a dynamic directed network with snapshots $G_0,\dots,G_T$
(Definition~\ref{def:dyn}). A \emph{matrix representation} $M$ assigns to every
snapshot $G_t$ a square matrix $M_t=M(G_t)$ of a fixed size $n$, for example, its
adjacency matrix or a community transition matrix.

\begin{definition}[Spectral braid]\label{def:sb}
Let $G$ be a dynamic directed network with snapshots $G_0,\dots,G_T$, and let $M$ be
a matrix representation. Connect consecutive matrices by straight lines,
\[
M(t+s)=(1-s)\,M_t+s\,M_{t+1},
\]
for $t=0,\dots,T-1$ and $s\in[0,1]$, and assume that the eigenvalues of $M(\tau)$
are distinct for every $\tau\in[0,T]$. The eigenvalues of $M(\tau)$ are then $n$
distinct points moving in $\C$, and the \emph{spectral braid of $G$ with respect to
$M$} is the braid they trace.
\end{definition}

The spectral braid is well defined for any matrix representation. In this paper, we
use the community transition matrix $P_t$ of Definition~\ref{def:ctm}, so that
$M_t=P_t$ and $n=K$. We choose it over a node-level matrix for three reasons. First,
each row of a node-level matrix is computed from the few edges of a single node, so its eigenvalues are noisy, whereas each row of the community transition matrix is
computed from the edges of all nodes in a community. Second, a network with $n$
nodes gives $n$ strands and costs $O(n^3)$ per step, while $K$ communities give at
most $K$ strands at $O(K^3)$ per step. Third, the
eigenvalues of a community transition matrix can be read in terms of how groups of
communities behave (Section~\ref{sec:interp}), which makes the crossings of the
braid interpretable.

Given a dynamic directed network $G$ and its community partition, the spectral braid
is computed in three steps.

\textbf{Step 1: Build the representation.} For each snapshot, compute the community
transition matrix $P_t$ by~\eqref{eq:ctm} (Fig.~\ref{fig:pipeline}a). A single pass
over the edges suffices, so the cost is $O(m_t)$ for $m_t$ edges.

\textbf{Step 2: Track eigenvalues.} On each interval $[t,t+1]$, evaluate the path
$P(t+s)=(1-s)P_t+sP_{t+1}$ of Definition~\ref{def:sb} at $S$ intermediate points and compute the eigenvalues at each point. Eigenvalues at consecutive points are matched
by minimum-cost assignment on their distances, which yields continuous
\emph{strands} (Fig.~\ref{fig:pipeline}b).

\textbf{Step 3: Read the braid.} Project the strands onto a fixed generic line
$\ell_\phi=\{\mathrm{Re}(e^{-i\phi}z)\}$ and number them by their order on this line.
Whenever the strands at positions $i,i+1$ exchange order, record $\sigma_i$ if the
exchange is counter-clockwise and $\sigma_i^{-1}$ otherwise
(Fig.~\ref{fig:pipeline}c). Finally, cancel adjacent pairs $\sigma_i\sigma_i^{-1}$.
The resulting word represents the spectral braid of $G$ with respect to $P$.

The definition does not require the network to return to its starting state. If
$P_T=P_0$, for example in a network with a daily or weekly cycle, the path is a
loop, the eigenvalues return to the set they started from, and the spectral braid is
a braid in the sense of Definition~\ref{def:braid}: it is the \emph{braid monodromy}
of the loop. Otherwise the eigenvalues end at different points and we call the
result an \emph{open} braid. Step~3 reads a word in both cases. For an open braid,
the word depends on the fixed line $\ell_\phi$, and a crossing that happens very close to the first or last snapshot is less reliable than one in the middle.

\subsection{Features from the Braid and the Strands}\label{sec:burau}
A braid word is easy to read but hard to give to a classifier. Its length varies, and
the same braid has many words. Simple summaries, such as the number of times each
generator occurs, ignore the order of the crossings. We therefore map the word $w$ to
its Burau matrix $\rho_t(w)$ (Fig.~\ref{fig:pipeline}d). The matrix is the same for
all words of the same braid, and it depends on the order of the letters because
matrix multiplication does. Networks may have different numbers of strands. We place
every braid in $B_N$ for a fixed $N$ by letting the missing strands stay idle, which
gives an $N\times N$ matrix for every network. We evaluate at $t=e^{i}$ and $t=0.5$
and use the real and imaginary parts of the entries as features. The value $t=e^{i}$
is transcendental, so the evaluation loses no information beyond the Burau
representation itself, and it lies on the unit circle, which keeps the entries
bounded for long words. The value $t=0.5$ adds a second, real-valued view of the
same braid.

The strands of Step~2 can also be described directly. For each strand, we compute
seven statistics: the mean and the standard deviation of its real and imaginary part, the signed area it encloses, its path length, and the correlation between its
real and imaginary parts. Strands are ordered by their real part at the first
snapshot. We call this descriptor \emph{tracked geometry}. It is the second output
of our method: it keeps metric information that the braid discards, and it does not
exist without eigenvalue tracking.

\textbf{Complexity.} Step 1 is $O(m_t)$ per snapshot and can be run in parallel or on
a stream. Steps 2--3 cost $O(S K^3)$ per snapshot, independent of the network size,
and the Burau matrix costs $O(N^3)$ per crossing. With $K$ fixed, the method is linear
in the number of edges.
\subsection{Theoretical Properties}\label{sec:props}
Before using spectral braids, three questions need to be answered: when does the braid
carry any information, what is special about the real matrices that networks
produce, and which properties of the network does the braid ignore? We answer them
in turn.

A braid needs eigenvalues that move in the plane, so it needs complex eigenvalues.
The first result identifies the networks that have none. Recall from
Definition~\ref{def:sb} that $P(\tau)$, $\tau\in[0,T]$, denotes the path that
connects the snapshots, with $P(t)=P_t$. For a fixed $\tau$, let $\pi$ be the
stationary distribution of the random walk with transition matrix $P(\tau)$, that
is, $\pi P(\tau)=\pi$, so that $\pi_a$ is the long-run share of time the walk spends
in community $a$. The quantity $\pi_a P(\tau)(a,b)$ is then the flow from $a$ to $b$
per step. The walk is called \emph{reversible} if this flow is balanced for every
pair of communities,
\begin{equation}\label{eq:db}
\pi_a\,P(\tau)(a,b)=\pi_b\,P(\tau)(b,a)\qquad\text{for all } a,b,
\end{equation}
a condition known as \emph{detailed balance}.

\begin{proposition}[No circulation, no braid]\label{prop:rev}
If $P(\tau)$ satisfies detailed balance~\eqref{eq:db} with $\pi_a>0$ for all $a$,
for every $\tau\in[0,T]$, then all eigenvalues of $P(\tau)$ are real and the
spectral braid has no crossings. This holds in particular when all snapshots $P_t$
are reversible with respect to the same stationary distribution.
\end{proposition}
\begin{proof}
Fix $\tau$ and let $\Pi$ be the diagonal matrix with entries $\pi_a$.
Equation~\eqref{eq:db} says that $\Pi P(\tau)$ is symmetric, so
$S=\Pi^{1/2}P(\tau)\,\Pi^{-1/2}$ is symmetric and has the same eigenvalues as
$P(\tau)$. These eigenvalues are therefore real for every $\tau$, and real
eigenvalues cannot exchange order without colliding. For the last statement, if
$P_t$ and $P_{t+1}$ satisfy~\eqref{eq:db} with the same $\pi$, then so does
$(1-s)P_t+sP_{t+1}$, because~\eqref{eq:db} is linear in $P$.
\end{proof}
A walk that is not reversible has, for some pair, more flow in one direction than
in the other. Because the total flow into a community equals the total flow out of
it, such an imbalance must continue around a cycle, and we call this net flow around
cycles \emph{circulation}. Community walks on \emph{undirected} graphs are reversible at every snapshot, so their snapshots have real eigenvalues. Hence, spectral braids are specific to directed networks, and among those, to networks with circulation. This tells us where the method applies and where it
does not (Section~\ref{sec:scope}).

Even with circulation, the braid of all eigenvalues is not the right object,
because $P_t$ has real entries.

\begin{proposition}[Real matrices]\label{prop:expsum}
For a closed path of real matrices with simple spectrum, the number of $\sigma$'s
minus the number of $\sigma^{-1}$'s in the braid is zero.
\end{proposition}
\begin{proof}
This number equals the winding number around $0$ of the discriminant
$\prod_{i<j}(\lambda_i-\lambda_j)^2$. For a real matrix, the discriminant is real, so
it cannot wind around $0$.
\end{proof}
The reason is that complex eigenvalues of a real matrix come in mirror pairs
$\lambda,\bar\lambda$. When one of them rotates counterclockwise, its mirror image
rotates clockwise, so every crossing in the upper half-plane is matched by a
crossing of the opposite sign in the lower half-plane. The braid of all eigenvalues
is therefore the braid of the upper half-plane together with its mirror image: it
has twice as many strands and carries no additional information, and summaries that
do not separate the two halves, such as the count above, vanish. Moreover, it is
often not defined, because real eigenvalues move on a line, where they collide. We
therefore use only one eigenvalue of each pair.

\begin{definition}[Upper-half-plane braid]\label{def:uhp}
$\beta_+$ is the spectral braid restricted to the eigenvalues with
$\mathrm{Im}\,\lambda>0$, defined while none of them crosses the real axis.
\end{definition}
Since the upper half-plane is itself a plane, $\beta_+$ is an ordinary
braid, and nothing is lost: the lower half is the mirror image of the upper half.
All braids in the rest of the paper are upper-half-plane braids.

Finally, a feature for dynamic networks should describe what happened and not how fast or how strongly it happened. The braid has this property by construction.

\begin{remark}[Invariance]\label{rem:inv}
A braid depends only on the topology of the motion. It is unchanged if the network
evolves faster or slower, burstily or uniformly, or along a deformed path, as long as
no two eigenvalues collide.
\end{remark}

\subsection{What eigenvalues and crossings mean}\label{sec:interp}
The eigenvalues of $P_t$ have a direct reading in a common situation. We describe it
for online communities, where a directed edge is a reply or a repost. Consider a group of $k\ge3$ communities that pass content around a loop. Look at one
interaction that starts in a community of the group. With probability $r$ it stays
inside the same community. With probability $f$ it goes to the next community of the
loop, with probability $g$ to the previous one, and with a small probability $x$ it
leaves the group. We summarize the group by two numbers:
\begin{definition}\label{def:ret}
The \emph{retention} of the group is the probability $r$ that an interaction
starting in one of its communities stays inside the same community.
\end{definition}

\begin{definition}\label{def:circ}
The \emph{circulation} of the group is the difference $c=f-g$ between the
probability $f$ that an interaction goes to the next community of the loop and the
probability $g$ that it goes to the previous one.
\end{definition}
In a network of online communities, we read the retention as the
\emph{echo-chamber level} of the group: how much each community talks to itself. We
read the circulation as its \emph{broadcast level}: how one-way the interactions
between its communities are. If $c$ is large, content is passed along a chain. If
$c=0$, the communities reply to each other back and forth, which is a conversation.

The transition matrix of the group is a $k\times k$ circulant block: every row
contains $r$ on the diagonal, $f$ in the next position, and $g$ in the previous one.
Its $k$ eigenvalues are therefore known in closed form,
\[
\lambda_j=r+f\,\omega^{j}+g\,\omega^{-j},\quad \omega=e^{2\pi i/k},\quad
j=0,\dots,k-1 .
\]
The eigenvalue $\lambda_0=1-x$ is real. The eigenvalues $\lambda_1$ and
$\lambda_{k-1}$ are complex conjugates, and for $f>g$ the first of them lies in the
upper half-plane. Using $f+g=1-x-r$ and $f-g=c$, it equals
\begin{equation}\label{eq:state}
\lambda_1 \;=\; \underbrace{r+(1-x-r)\cos\tfrac{2\pi}{k}}_{\text{echo-chamber level}}
\;+\; i\,\underbrace{c\,\sin\tfrac{2\pi}{k}}_{\text{broadcast level}} .
\end{equation}
For $k=3$ this is $\tfrac32 r-\tfrac{1-x}{2}+i\tfrac{\sqrt3}{2}c$, and for $k=4$ it
is $r+ic$. In both cases, $\lambda_1$ is the only eigenvalue of the group in the
upper half-plane, and we call it \emph{the eigenvalue of the group}. Its real part
grows with the retention $r$, and its imaginary part is proportional to the
circulation $c$. So the eigenvalue of the group is a point that moves right when the
group becomes more of an echo chamber and up when it broadcasts more. A group of
five or more communities has further complex eigenvalues $\lambda_2,\lambda_3,\dots$
and contributes more than one point.

Now consider a network made of several such groups. Its transition matrix $P_t$
contains one block per group. If the interactions that leave a group are spread
evenly over the communities of the group they arrive in, the eigenvalue of each
group is also an eigenvalue of $P_t$. The upper half-plane then contains one point
per group, and each point moves according to the echo-chamber level and the
broadcast level of its own group. In this setting, a crossing of the braid has a
plain meaning. It is the moment one group becomes a stronger echo chamber than
another. Its sign says which of the two was broadcasting more at that moment:
$\sigma_i^{-1}$ if it was the group whose echo-chamber level is rising, and
$\sigma_i$ if it was the group whose echo-chamber level is falling. The braid word
is therefore a transcript of how the groups overtook each other. The sign is
relative: it compares the two groups involved and does not refer to a fixed level.

We close this section with an example that brings the preceding notions together.

\begin{example}\label{ex:uhp}
Consider a network with six communities that form two groups, $A=\{1,2,3\}$ and
$B=\{4,5,6\}$. At the first snapshot, its community transition matrix, rounded to
three decimals, is
\[
\setlength{\arraycolsep}{3pt}
P_0=\begin{pmatrix}
.423 & .361 & .165 & .017 & .017 & .017\\
.165 & .423 & .361 & .017 & .017 & .017\\
.361 & .165 & .423 & .017 & .017 & .017\\
.017 & .017 & .017 & .497 & .325 & .129\\
.017 & .017 & .017 & .129 & .497 & .325\\
.017 & .017 & .017 & .325 & .129 & .497
\end{pmatrix}.
\]
Group $A$ has retention $0.423$ and circulation $0.361-0.165=0.196$; group $B$ has
retention $0.497$ and circulation $0.196$. The eigenvalues of $P_0$ are
\[
1,\qquad 0.9,\qquad 0.16\pm0.17\,i,\qquad 0.27\pm0.17\,i .
\]
By~\eqref{eq:state}, $0.16+0.17\,i$ is the eigenvalue of group $A$ and
$0.27+0.17\,i$ that of group $B$; the two eigenvalues below the real axis are their
mirror images.

Now let group $B$ stay as it is, while the retention of group $A$ rises to $0.570$
and returns, and its circulation rises to $0.294$, falls to $0.098$, and returns.
The network ends in the state it started from, so the path is closed. The
eigenvalue of $A$ then goes once around the eigenvalue of $B$, clockwise, and its
mirror image goes around the mirror image of $B$, counter-clockwise
(Fig.~\ref{fig:uhp}a).

The upper-half-plane braid is $\sigma_1^{-1}\sigma_1^{-1}$ (Fig.~\ref{fig:uhp}b),
and since the path is closed, it is the braid monodromy of this loop. Its first
letter records that $A$ became a stronger echo chamber than $B$ while broadcasting
more than $B$; its second letter records that $B$ overtook $A$ again while
broadcasting more than $A$. The mirror images give $\sigma_1\sigma_1$
(Fig.~\ref{fig:uhp}c), which adds no information, and in the braid of all eigenvalues, the two kinds of crossings cancel in the count of
Proposition~\ref{prop:expsum}.
\end{example}

\begin{figure}[t]
\centering
\includegraphics[width=.9\textwidth]{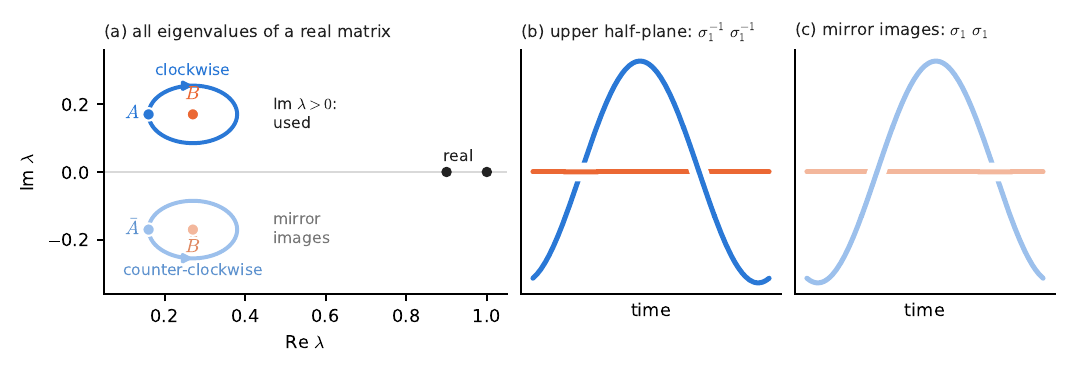}
\caption{Example~\ref{ex:uhp}. (a) The six eigenvalues of the transition matrix as
group $A$ changes and returns to its starting state; group $B$ does not change.
$A$ and $B$ mark the eigenvalues of the two groups at the first snapshot, and
$\bar A$, $\bar B$ their mirror images. Black: real eigenvalues.
(b) Braid of the eigenvalues in the upper half-plane. (c) Braid of their mirror
images.}
\label{fig:uhp}
\end{figure}
\section{Experiments}\label{sec:exp}
In this section, we evaluate spectral braids on dynamic network classification. The experiments are designed to show whether the braid can tell how a network evolved when its first and last snapshots carry no information, whether the Burau matrix
captures the order in which events happened, whether the features still work on
networks that look different from the training networks, how many edges per snapshot
the method needs, and how fast it runs on large networks.

\subsection{Datasets}\label{sec:bench}
To evaluate the method, we need dynamic, directed networks with known histories. As
discussed in Section~\ref{sec:scope}, the real datasets we examined do not carry
rotating circulation, so we generate two synthetic datasets of online communities, in
which a directed edge is a reply or a repost. Table~\ref{tab:dataset_summary}
summarizes their properties.

\begin{table}[t]
\centering
\caption{Summary of the datasets used in the experiments. Ranges are over the
networks of a dataset, with the mean in parentheses.}
\label{tab:dataset_summary}
\begin{adjustbox}{max width=\textwidth}
\begin{tabular}{|c|c|c|c|c|}
\hline
\cellcolor{gray!25}\textbf{Dataset} & \cellcolor{gray!25}\textbf{\# Networks} &
\cellcolor{gray!25}\textbf{\# Classes} & \cellcolor{gray!25}\textbf{\# Nodes} & \cellcolor{gray!25}\textbf{\# Snapshots} \\
\hline
History & 600 & 4 & 195--505 (343) &  50--152 (101)  \\
Order   & 600 & 8 & 184--506 (347) &  72--133 (101)  \\
\hline
\end{tabular}
\end{adjustbox}
\end{table}

\subsubsection{Network Generation}
Each sample is one dynamic directed network together with one label. To build a
network, we first draw the number of groups, between 4 and 6, and for each group, the number of its communities, 3 or 4. A network, therefore, has 12 to 24 communities.
Each community receives between 10 and 30 users, with an average of 20. The communities of a
group pass content around a cycle, as in Section~\ref{sec:interp}. We use groups of
this form because complex eigenvalues arise only from net flow around cycles
(Proposition~\ref{prop:rev}).

At every snapshot, each group has a retention $r$ and a circulation $c$. An interaction that starts in a community stays in that community with probability
$r$, goes to the next or to the previous community of its group with probabilities
determined by $c$, and leaves the group with probability 5\%, in which case its target is spread evenly over the other groups. To avoid an exact block structure, we
mix the resulting matrix with a random transition matrix, with weights 97\% and 3\%.

At the first snapshot, the groups have equally spaced echo-chamber levels and the
same broadcast level. The network then goes through a sequence of events determined
by its label, described below for each dataset. In an event, two groups with
neighboring echo-chamber levels exchange their levels, while one of them
temporarily broadcasts more and the other less. The properties of an event that are
not part of the label are drawn at random for every event: its duration (10 to 26
snapshots), the pause before the next event (2 to 9 snapshots), whether it proceeds
at a steady or an irregular speed, and the size of the temporary change in broadcast
level, which varies by a factor of about 3. In addition, the broadcast levels of all
groups drift slowly at random. Since these properties are drawn independently of the
label, they carry no information about it.

Finally, the edges are sampled. In every snapshot, each user sends 10 interactions,
and the community that receives each interaction is drawn from the probabilities of
that snapshot. The observed community transition matrices are therefore noisy
versions of the ideal ones. The methods receive the communities and the edges only: the communities are listed in random order, and no method is told which communities
form a group.

\subsubsection{History Dataset}
In every network, one group of communities rises from the lowest echo-chamber rank
to the highest, overtaking the other groups one at a time. Every network, therefore,
starts in the same kind of state and ends in the same kind of state, and the label
says how the rising group became the strongest echo chamber. Each time the rising
group overtakes another group, one of the two is broadcasting more than the other,
and the four classes differ in which one. In \emph{Broadcast-driven} networks, the
rising group is the broadcaster at every step: it becomes an echo chamber while
passing content one-way. In \emph{Conversation-driven} networks, the rising group is in conversation
at every step, and the group it overtakes is the broadcaster. In \emph{Mixed,
broadcast first} and \emph{Mixed, conversation first} networks, the two styles
alternate from one step to the next, starting with broadcast or with conversation.

\subsubsection{Order Dataset}
This dataset tests whether the order of events can be recovered. Every network
contains the same four events. The two lowest-ranked groups exchange ranks twice,
once with the group that moves up broadcasting more and once with it broadcasting
less, and so do the second- and third-ranked groups. The eight classes differ only in the order in
which these four events happen, so any summary that counts events without their
order is identical for all classes.

\begin{figure}[t]
\centering
\includegraphics[width=\textwidth]{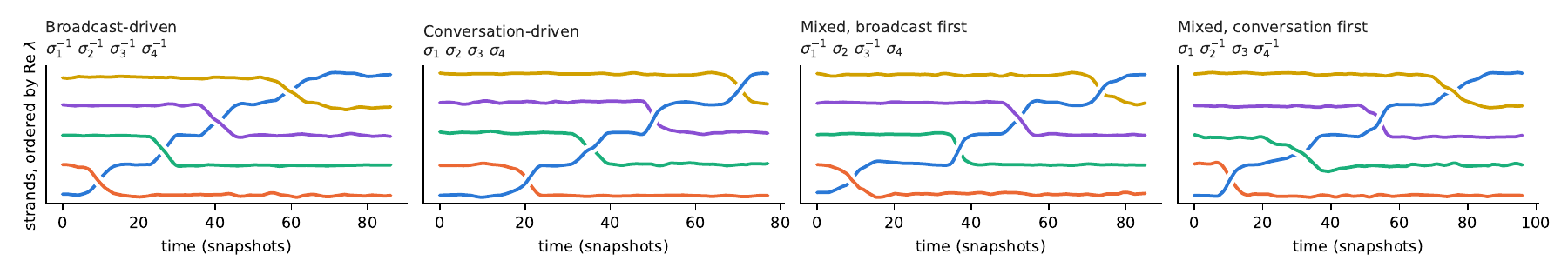}
\caption{One random network per label of the History dataset, with five strands: the
tracked eigenvalues drawn over time and the braid word read from the data. In all four, the blue strand goes from the bottom to the top; the labels differ only in whether it passes behind or in front of the others.}
\label{fig:classes}
\end{figure}

Both datasets are designed so that the classes differ only in the path between the
same start and end states, which is the situation spectral braids are meant for; they do not measure how often this situation occurs in practice.

\subsection{Evaluation}
We evaluate all methods on dynamic network classification and report accuracy.

\subsubsection{Compared Methods}
All methods receive the same input, the sequence of community transition matrices
$P_t$, and all use a random forest with 300 trees. Three baselines use the
eigenvalues without tracking them. \emph{Pooled spectra} puts the eigenvalues of all
snapshots together and summarizes them by quantiles and a histogram, ignoring time
order. \emph{Start and end} uses the eigenvalues of the first three and of the last
three snapshots. \emph{Sorted series} sorts the eigenvalues in each snapshot, which
gives a time series, and summarizes each series by the seven statistics of
Section~\ref{sec:burau}. Our method tracks the eigenvalues and describes the strands
in three ways. \emph{Tracked geometry} computes the same seven statistics on the
tracked strands. \emph{Exponent sums} counts, for each generator $\sigma_i$, how
often it occurs in the braid word and with which total sign. \emph{Burau matrix} uses
the features of Section~\ref{sec:burau} with $N=6$.

\subsubsection{Evaluation Settings}
We evaluate the methods in four settings.

In the first setting, we measure classification accuracy on each dataset with
5-fold cross-validation, repeated 5 times with different splits, and report the mean
and standard deviation over the 25 runs.

In the second setting, we test whether a classifier trained on one kind of network
works on another. For this purpose, we generate two additional versions of each
dataset, which differ only in how the events of
Section~\ref{sec:bench} unfold over time. In \emph{calm} networks, the events are
regular: every event lasts 16 to 20 snapshots and proceeds at constant speed, the
temporary change in broadcast level has nearly the same size in every event, and the
groups do not drift. In \emph{turbulent} networks, the events are irregular: an
event lasts between 10 and 26 snapshots and proceeds at an uneven speed, the
temporary change in broadcast level varies by a factor of about 3, and the
broadcast levels of all groups drift. The turbulent version follows the same
generation procedure as the main datasets. For each dataset, we generate 400 calm
and 400 turbulent networks, with the same number of networks in every class. We
train on all networks of one kind and test on all networks of the other, in both
directions.

In the third setting, we study how the accuracy depends on the amount of data in a
snapshot. We regenerate the History dataset with 10, 20, 40, 80, and 160 users per
community on average, which gives about 1{,}700, 3{,}400, 7{,}000, 13{,}700, and
28{,}000 edges per snapshot, with 200 networks at each size. For the Order dataset,
we do the same with 10, 40, and 80 users per community, with 400 networks at each
size. At every size, we measure the accuracy with a 5-fold cross-validation repeated
3 times.

In the fourth setting, we measure the running time of the two parts of the method,
building the community transition matrices and tracking the eigenvalues, on one CPU
thread.

\begin{table}[t]
\centering
\caption{Classification accuracy (\%). CV: 5-fold cross-validation (mean $\pm$ std).
C$\to$T and T$\to$C: trained on calm networks and tested on turbulent ones, and the
reverse. Best value per column is in bold; the second-best is underlined.}
\label{tab:acc}
\begin{adjustbox}{max width=\textwidth}
\begin{tabular}{|l|ccc|ccc|}
\hline
 & \multicolumn{3}{c|}{\cellcolor{gray!60}\textbf{History}} &
   \multicolumn{3}{c|}{\cellcolor{gray!60}\textbf{Order}} \\
\cline{2-7}
 & \cellcolor{gray!25}\textbf{CV} & \cellcolor{gray!25}\textbf{C$\to$T} & \cellcolor{gray!25}\textbf{T$\to$C}
 & \cellcolor{gray!25}\textbf{CV} & \cellcolor{gray!25}\textbf{C$\to$T} & \cellcolor{gray!25}\textbf{T$\to$C} \\
\hline
\multicolumn{7}{|l|}{\cellcolor{gray!10}\emph{Baselines (eigenvalues not tracked)}} \\
\hline
\cellcolor{gray!25}\textbf{Pooled spectra} & $28.2\pm3.9$ & 23.5 & 25.8 & $12.6\pm2.5$ & 13.8 & 11.5 \\
\cellcolor{gray!25}\textbf{Start and end}  & $26.6\pm3.4$ & 21.8 & 25.2 & $12.2\pm2.4$ & 11.2 & 13.0 \\
\cellcolor{gray!25}\textbf{Sorted series}  & $44.5\pm4.2$ & 41.2 & 47.8 & $12.4\pm2.8$ & 10.5 & 10.5 \\
\hline
\multicolumn{7}{|l|}{\cellcolor{gray!10}\emph{Ours (tracked eigenvalues)}} \\
\hline
\cellcolor{gray!25}\textbf{Tracked geometry} & $94.7\pm2.0$ & 91.2 & 96.2 & $\mathbf{93.2\pm2.3}$ & \textbf{92.8} & \textbf{92.0} \\
\cellcolor{gray!25}\textbf{Exponent sums}    & \underline{$96.9\pm1.3$} & \textbf{96.8} & \textbf{97.2} & $12.3\pm1.8$ & 11.8 & 10.0 \\
\cellcolor{gray!25}\textbf{Burau matrix}     & $\mathbf{97.1\pm1.4}$ & \underline{93.5} & \underline{97.0} & \underline{$90.8\pm2.7$} & \underline{88.8} & \underline{86.5} \\
\hline
\end{tabular}
\end{adjustbox}
\end{table}

\subsection{Results}

\subsubsection{Classification}
Table~\ref{tab:acc} reports the accuracy of all methods on the two datasets.
Overall, the results show a clear gap between the baselines and our method: once the
eigenvalues are tracked, both tasks become solvable, and without tracking, neither is.

On the History dataset, pooled spectra and start-and-end features are at chance
(28.2\% and 26.6\%). This is expected, because the labels differ only in the path
between the same two states. The sorted series reaches 44.5\%. It uses time, but
sorting mixes up which eigenvalue is which whenever two of them pass each other.
With tracked eigenvalues, all three of our descriptors reach 94.7\% to 97.1\%, more
than 50 points above the best baseline, and the Burau matrix is the highest. The
differences among our three descriptors are within their standard deviations.
Fig.~\ref{fig:classes} shows the braid read from one network of each class.

On the Order dataset, all three baselines are at chance, and the two braid
descriptors separate. Exponent sums fall to chance (12.3\%), because every network
contains the same four crossings. The Burau matrix reaches 90.8\%, because the
product of the generator matrices depends on their order. Tracked geometry reaches
93.2\%, level with the Burau matrix within the standard deviations. This indicates
that the Burau matrix is the braid descriptor to use: it matches simple counts when
order does not matter and keeps working when it does.

\subsubsection{Transfer to Different-Looking Networks}
The transfer columns of Table~\ref{tab:acc} are close to the cross-validation
columns for our descriptors. Burau features keep 93.5\% and 97.0\% on the History
dataset, and 88.8\% and 86.5\% on the Order dataset, and tracked geometry keeps
91.2\% to 96.2\%. For the braid, this agrees with the invariance remark of
Section~\ref{sec:props}: a braid does not depend on the speed or the shape of the
motion. The baselines remain at the levels they had during cross-validation.

\begin{figure}[t]
\centering
\includegraphics[width=0.5\textwidth]{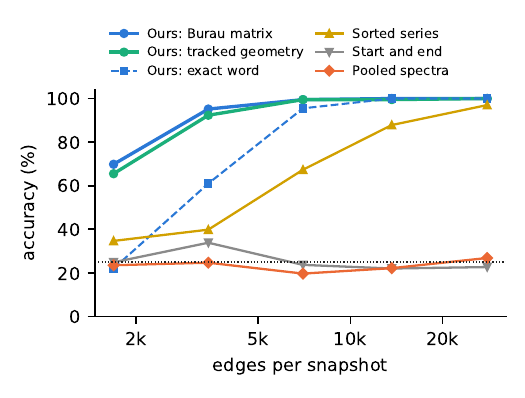}
\caption{Accuracy on the History dataset against the number of edges per snapshot,
for our descriptors and the three baselines. Exact word: the share of networks whose
braid word equals the planted one, with no training. The dotted line is chance.}
\label{fig:sweep}
\end{figure}

\subsubsection{Sensitivity to the Number of Edges}
The matrix $P_t$ is estimated by counting the edges of one snapshot. With few edges,
the estimate is noisy and the eigenvalues jitter, and when the jitter is as large as
the distance between two eigenvalues, the tracking records a crossing that did not
happen or misses one that did. Fig.~\ref{fig:sweep} shows how the accuracy depends
on the number of edges per snapshot.

Our descriptors are above the baselines across all densities. Burau features reach
69.8\% at 1{,}700 edges, 95.2\% at 3{,}400, and 99.5\% at 7{,}000, and tracked
geometry reaches 65.5\%, 92.3\%, and 99.5\%. The values at 3{,}400 edges differ
slightly from those in Table~\ref{tab:acc} because they are measured on a separate
sample of 200 networks. The sorted series needs far more data
and reaches 97.0\% only at 28{,}000 edges. Pooled spectra and start-and-end features
stay near chance at every density, because more data does not add the missing
information about the path.

The dashed line shows how often the word read from the data is exactly the planted
word: 61\% at 3{,}400 edges and 100\% from 13{,}700 edges. The classifier is more
tolerant than exact matching, since many networks whose words contain an error are
still classified correctly. On the Order dataset with sparse data, tracked geometry
is the more accurate of our two descriptors (66.6\% against 55.2\% at 1{,}650
edges), because averaging over a strand absorbs noise while one wrong crossing
changes the whole Burau matrix; from about 6{,}800 edges both exceed 98.9\%. These
results suggest that the method needs a few thousand edges per snapshot for about 20
communities.

\subsubsection{Running Time}
On one CPU thread and 21 communities, building $P_t$ takes 0.6 ms for $10^5$ edges,
8.5 ms for $10^6$, and 144 ms for $10^7$, in proportion to the number of edges.
Tracking the eigenvalues takes about 1 ms per snapshot whatever the size of the
network. A network with ten million edges per snapshot and 100 snapshots is
therefore processed in about 15 seconds.

\section{Applicability to Real Networks}\label{sec:scope}
We applied the method to four real dynamic directed networks, and none of them gave
a non-trivial braid that was reproducible. The first is the e-mail network of a
research institution~\cite{paranjape2017}, where nodes are people; we used its six
largest communities, found by the Louvain method, and 4-week snapshots. The second
is international trade from 1950 to 2009~\cite{barbieri2009}, where nodes are
countries; we grouped them into six world regions and used yearly snapshots. The
third is the Reddit hyperlink network~\cite{kumar2018}, with 858{,}454 links between
subreddits; we used 10 Louvain communities of subreddits and 3-month snapshots. In
these three networks, the eigenvalues of $P_t$ are real or almost real, so there
are no strands to braid. On Reddit, for example, a snapshot has, on average, 0.03
pairs of complex eigenvalues out of 5 possible.

The fourth is the Bay Area Bike Share system~\cite{babs2015}, with 628{,}197 trips
between stations in San Francisco from 2013 to 2015; we grouped the stations into
eight zones by location. Here, the eigenvalues of $P_t$ are also mostly real. To give
one-way flow more weight, we replaced $P_t$ by the matrix $S+\theta K$, where $S$
and $K$ are the symmetric and antisymmetric parts of the matrix of trip counts
between zones, and $\theta>1$ amplifies the antisymmetric part. With this matrix,
three eigenvalues stay complex throughout the morning rush and throughout the
evening rush, so a braid is defined for each. In both periods, the braid is the
trivial one, and it remains trivial when the days are split at random into two
halves.

These results are what Proposition~\ref{prop:rev} predicts, and they point to a
second condition. In the e-mail, trade and Reddit networks, the flows between
communities are close to balanced, so the random walks are close to reversible and
their eigenvalues are real. The bike data is not balanced: in the morning rush, some
zones send out many more trips than they receive, and in the evening, the direction
reverses. Within a rush period, however, the same zones send to the same zones
throughout, and only the strength of the flow rises and falls: a single fixed matrix
of net flows accounts for 92\% of their variation in the morning and 94\% in the
evening. The eigenvalues then move away from the real axis and return along the
paths they came on, so they do not pass around each other. A non-trivial spectral
braid therefore needs more than one-way flow: the \emph{direction} of the net flow
must change over time among at least three groups, for example from $a\to b$ to
$b\to c$ to $c\to a$, and not only its strength.

\section{Conclusion}
In this paper, we introduced spectral braids, a topological methodology for dynamic
directed networks, and, to the best of our knowledge, the first definition of braids
and braid monodromy on graph data. The method aggregates each snapshot into a
community transition matrix, tracks its complex eigenvalues over time, and records
how they pass around each other as a braid. The braid is defined for any observation
window, costs time linear in the number of edges, and becomes a fixed-size feature
vector through the Burau representation. We showed that the braid is non-trivial only when the network carries circulation, introduced the upper-half-plane braid, which removes the mirror-image redundancy of
real matrices, and gave the eigenvalues and the
crossings a concrete meaning in terms of how groups of communities behave. On two
synthetic datasets of online communities, in which the classes differ only in the
path between the same starting and ending states, the proposed features reach
97.1\% accuracy on four classes, against 44.5\% for the best features that do not
track eigenvalues. When the classes contain the same events in a different order,
the Burau features reach 90.8\%, while counts of crossings stay at chance. The
features remain accurate on networks that evolve less regularly than the training
networks.

In future work, we plan to apply spectral braids to real systems in which the
direction of the net flow changes among groups over time, to define the braid
without a given community partition from the eigenvalues of the node-level matrix,
to use the eigenvectors of the tracked eigenvalues to obtain node-level features, to
study braid invariants and representations beyond the Burau matrix, and to use braid
features as inputs to temporal graph neural networks.

\bibliographystyle{unsrt}
\bibliography{references}

\end{document}